\documentclass{article}

\usepackage{microtype}
\usepackage{graphicx}
\usepackage{subcaption}
\usepackage{booktabs} 

\usepackage{hyperref}

\usepackage[accepted]{icml2026}
\setcitestyle{numbers,square,sort&compress,comma}

\usepackage{amsmath}
\usepackage{amssymb}
\usepackage{mathtools}
\usepackage{amsthm}
\usepackage{enumitem}

\renewcommand{\vec}[1]{\boldsymbol{#1}}

\usepackage[capitalize,noabbrev]{cleveref}

\theoremstyle{plain}
\newtheorem{theorem}{Theorem}          
\newtheorem{proposition}{Proposition}  

\theoremstyle{definition}
\newtheorem{definition}{Definition}    
\newtheorem{assumption}{Assumption}    

\usepackage[textsize=tiny]{todonotes}

\icmltitlerunning{Socially-Weighted Alignment: A Game-Theoretic Framework for Multi-Agent LLM Systems}

\begin{document}

\twocolumn[
\icmltitle{Socially-Weighted Alignment: A Game-Theoretic Framework for \\ Multi-Agent LLM Systems}

\vskip 0.15in

\begin{center}
{\fontsize{11.25pt}{13.5pt}\bfseries\selectfont
Furkan Mumcu \quad
Yasin Yilmaz
}

\vskip 0.1in  

{\fontsize{11.25pt}{13.5pt}
University of South Florida
}

\vskip 0.08in 

{\tt\small
\{furkan, yasiny\}@usf.edu \quad
}
\end{center}

\vskip 0.25in
]



\printAffiliationsAndNotice{}  

\begin{abstract}

Deploying large language model (LLM) agents in shared environments introduces a fundamental tension between individual alignment and collective stability: locally rational decisions can impose negative externalities that degrade system-level performance. We propose Socially-Weighted Alignment (SWA), a game-theoretic framework that modifies inference-time decision making by interpolating between an agent's private objective and an estimate of group welfare via a social weight $\lambda\in[0,1]$. In a shared-resource congestion game with $n$ agents and congestion severity $\beta$, we show that SWA induces a critical threshold
$\lambda^*=(n-\beta)/(n-1)$ above which agents no longer have marginal incentive to increase demand under overload, yielding a phase transition from persistent congestion to stable operation near capacity. We further provide an inference-time algorithmic instantiation of SWA that does not require parameter updates or multi-agent reinforcement learning, and use a multi-agent simulation to empirically validate the predicted threshold behavior.
\end{abstract}

\section{Introduction}

LLM agents are increasingly deployed not as isolated assistants but as interacting components in shared environments. Examples include multi-agent tool-use pipelines, collaborative editing systems, and agent teams that coordinate over shared memory, budgets, or external state. In these settings, standard alignment at inference time typically optimizes each agent's local objective, such as instruction following or task completion for the current user request. However, when agents share limited resources, locally rational actions can impose negative externalities on others. In a collaborative coding workflow, for example, one agent’s ``efficient" refactor may generate costly merge conflicts for teammates; in a tool-use swarm, aggressive retries by individual agents can trigger global API rate limits, halting the entire system.
The resulting dynamics resemble classic social dilemmas: individually optimal behavior can degrade collective performance through congestion, duplication, or conflict.

Most existing multi-agent LLM systems orchestrate collaboration through procedural design choices such as role assignment, turn-taking, debate, voting, or reflection, as implemented in popular multi-agent frameworks.
In parallel, recent work evaluates LLM agents in game-theoretic settings, including social dilemmas and congestion-style games, primarily to characterize emergent strategic behavior or to measure prosocial tendencies.
In contrast, our goal is incentive-level orchestration. Rather than prescribing an interaction protocol, we modify the inference-time objective each agent optimizes by interpolating between a private score and a group-welfare score with a single parameter $\lambda$.
This lets us treat the resulting multi-agent system as an induced game, derive a closed-form stability threshold in a shared-resource congestion environment, and test a falsifiable prediction about how the transition shifts with environment parameters.

This paper studies a simple question: can we modify inference-time decision making so that agents internalize a controlled amount of group welfare, while preserving decentralization and avoiding costly training-time interventions. We propose Socially-Weighted Alignment (SWA), a framework in which each agent optimizes an effective utility that interpolates between a private score and a group score. The interpolation coefficient $\lambda$ governs a continuum of behaviors from individually aligned play ($\lambda=0$) to fully welfare-driven play ($\lambda=1$). Conceptually, $\lambda$ acts as a knob for trading off individual goal pursuit against system-level stability.

We formalize SWA as an induced game and analyze it in a shared-resource congestion environment. In this setting, $n$ agents choose demands that collectively determine an aggregate load $X_t$ relative to a capacity $C$, and overload incurs a congestion penalty scaled by a severity parameter $\beta$. Our analysis yields a sharp prediction: there exists a critical social weight $\lambda^*=(n-\beta)/(n-1)$ such that, in the overloaded regime, agents lose the marginal incentive to increase demand when $\lambda\ge \lambda^*$. This provides a mechanistic explanation for how socially weighted objectives can produce a coordination-driven phase transition from persistent overload to stable operation near capacity.

To operationalize SWA for LLM agents, we introduce an inference-time protocol that evaluates a finite set of candidate actions using natural-language scoring rules that approximate the private and group components of the SWA objective. The protocol does not require parameter updates, centralized control, or multi-agent reinforcement learning, and can be applied across model families.

We empirically evaluate SWA in a multi-agent simulation of a shared-resource congestion environment. We measure overload rate (OR) and welfare improvement relative to the $\lambda=0$ baseline. Across multiple base models, increasing $\lambda$ induces a sharp reduction in overload and a corresponding improvement in welfare, consistent with the phase-transition behavior predicted by SWA. We further vary congestion severity as an ablation and observe that the transition shifts in the direction implied by the theoretical threshold, and that higher severity increases the welfare cost of remaining in the overloaded regime.

In summary, this work advances a game-theoretic view of agentic LLM system design by treating orchestration as inference-time incentive design. We contribute:

\begin{itemize}

    \item A general welfare-weighted inference-time objective, SWA, that reshapes the induced game by interpolating private utility and estimated group welfare without parameter updates,
    \item An analytically tractable characterization of system stability in a canonical shared-resource congestion game, including a closed-form critical threshold $\lambda^*=(n-\beta)/(n-1)$ that predicts a sharp transition from persistent overload to stable operation near capacity,
    \item Systematic empirical validation that this threshold phenomenon is robust across multiple small and efficient base model families and shifts predictably under a congestion-severity ablation.
\end{itemize}

\section{Mathematical Framework}

We formalize interactions among Large Language Model (LLM) agents as a multi-agent decision process and introduce {Socially-Weighted Alignment (SWA)}, a family of \emph{effective objectives} that interpolate between purely self-interested instruction following and group-aware behavior.
Importantly, our setting is {not multi-agent reinforcement learning (MARL)}. We do not train policies through repeated rollouts or reward optimization. Instead, we study {inference-time} decision rules for fixed, pretrained LLM agents operating in shared tool-using environments.

\subsection{Preliminaries: A Multi-Agent Agentic Environment}

We consider a multi-agent environment $\mathcal{G} = \langle N, \mathcal{S}, \mathcal{A}, \mathcal{P}, \mathcal{R}, \gamma \rangle$, where:

\begin{itemize}[leftmargin=*]
    \item \textit{Agents ($N$):} $N=\{1,\dots,n\}$ is a set of LLM agents (e.g., roles within an agentic workflow).
    \item \textit{State Space ($\mathcal{S}$):} $s_t \in \mathcal{S}$ denotes the global state at time $t$. In agentic LLM systems, $s_t$ may include the shared conversation/context, tool outputs, and external state (e.g., repository state, shared memory). Agents may only observe a \emph{projection} of $s_t$ (partial observability).
    \item \textit{Action Space ($\mathcal{A}$):} $\mathcal{A}=\mathcal{A}_1 \times \dots \times \mathcal{A}_n$ is the joint action space. An action $a_i \in \mathcal{A}_i$ is a message and/or tool invocation produced by agent $i$ at inference time.
    \item \textit{Transition Dynamics ($\mathcal{P}$):} The environment evolves via $\mathcal{P}(s_{t+1}\mid s_t,\vec{a}_t)$, where $\vec{a}_t=(a_{1,t},\dots,a_{n,t})$. In practice, $\mathcal{P}$ captures how the shared workspace changes (e.g., files edited, memory updated, context length consumed).
    \item \textit{Per-agent Reward / Score ($\mathcal{R}$):} $R_i:\mathcal{S}\times\mathcal{A}\rightarrow\mathbb{R}$ is the intrinsic objective for agent $i$ (e.g., task progress, correctness, instruction adherence). We treat $R_i$ as an abstract score; in agentic LLM systems it is often \emph{not directly observed} and must be approximated by heuristics or model-based evaluators.
    \item \textit{Discount Factor ($\gamma$):} $\gamma \in [0,1)$ controls the horizon.
\end{itemize}

Each agent selects actions using an inference-time policy $\pi_i(a_i \mid o_{i,t})$ based on its observation $o_{i,t}$ (e.g., current context and tool outputs). Unless explicitly stated, we assume model weights are fixed during deployment.

\subsection{Baseline: Individually Aligned Agents}

In standard agentic deployments, agents are aligned individually. Agent $i$ selects actions to maximize its own expected discounted return:
\begin{equation}
    V_i^{\pi}(s) = \mathbb{E}_{\pi}\left[\sum_{t=0}^{\infty}\gamma^t R_i(s_t,\vec{a}_t)\,\bigg|\, s_0=s \right].
\end{equation}

A joint policy $\vec{\pi}^*=(\pi_1^*,\dots,\pi_n^*)$ is a {Nash equilibrium (NE)} if no agent can improve its own value by unilateral deviation:
\begin{equation}
    V_i^{(\pi_i^*,\pi_{-i}^*)}(s)\ge V_i^{(\pi_i,\pi_{-i}^*)}(s),\quad \forall i\in N,\ \forall \pi_i,
\end{equation}
where $\pi_{-i}^*$ denotes all policies except $i$. 

\noindent \textit{Critique:} In shared-resource or shared-state workflows, maximizing $R_i$ can impose negative externalities on other agents (e.g., context-window depletion, redundant tool calls, inconsistent edits). Consequently, equilibria under individually aligned objectives can be systemically inefficient.

\subsection{Socially-Weighted Alignment (SWA)}

We introduce a simple family of {effective objectives} that internalize externalities by blending self-interest with a group objective.

\textit{Global welfare:} We define the global welfare as the aggregate intrinsic score,
\begin{equation}
    W(s,\vec{a})=\sum_{j=1}^{n}R_j(s,\vec{a}).
\end{equation}

\textit{Social alignment coefficient:}
Each agent $i$ is assigned a social alignment coefficient $\lambda_i\in[0,1]$ that controls the trade-off between its own objective and group welfare.

\begin{definition}[SWA social utility]
The SWA effective utility for agent $i$ is
\begin{equation}
    U_i^{\lambda}(s,\vec{a})
    =(1-\lambda_i)\,R_i(s,\vec{a})+\lambda_i\,\frac{1}{n}W(s,\vec{a}).
\end{equation}
\end{definition}

This defines a spectrum of behaviors:
\begin{enumerate}[leftmargin=*]
    \item \textit{Pure self-interest ($\lambda_i=0$):} recovers individually aligned deployment.
    \item \textit{Pure group-interest ($\lambda_i=1$):} the agent optimizes average welfare.
    \item \textit{Mixed motive ($0<\lambda_i<1$):} balances task progress with system stability.
\end{enumerate}

\subsection{Equilibrium of the SWA-Induced Game}

SWA does not introduce a new equilibrium concept; rather, it induces a new game with utilities $\{U_i^{\lambda}\}_{i=1}^n$.
We study standard equilibrium behavior under these transformed utilities.

\begin{definition}[SWA Nash equilibrium]
A joint policy $\vec{\pi}^*$ is an {SWA equilibrium} if it is a Nash equilibrium with respect to the SWA returns:
\begin{equation}
    J_i^{\lambda}(\pi_i^*,\pi_{-i}^*) \ge J_i^{\lambda}(\pi_i,\pi_{-i}^*)
    \quad \forall i\in N,\ \forall \pi_i,
\end{equation}
where
\begin{equation}
    J_i^{\lambda}(\vec{\pi})=\mathbb{E}_{\vec{\pi}}\left[\sum_{t=0}^{\infty}\gamma^t U_i^{\lambda}(s_t,\vec{a}_t)\right].
\end{equation}
\end{definition}

\textit{Interpretation: ``price of alignment.''}
Varying $\lambda_i$ traces a trade-off curve between (i) an agent's task-specific score and (ii) multi-agent system health (e.g., fewer collisions, less wasted context, higher joint task success). In Section~4, we introduce inference-time procedures that approximate $U_i^{\lambda}$ in agentic LLM systems where $R_i$ and $W$ are not directly observable.

\section{Theoretical Analysis}

In this section, we analyze how the Social Alignment Coefficient $\lambda$ affects stability in a shared-resource setting.
We model agentic externalities (e.g., shared context-window budget, shared tool rate limits) via a {resource congestion game} \cite{rosenthal1973congestion}, a continuous variant of the Tragedy of the Commons \cite{hardin1968tragedy}.

\subsection{Problem Setup and Assumptions}

We consider $n$ agents competing for a shared resource with capacity $C>0$.
Let $x_i \in [0, x_{\max}]$ denote the resource demand of agent $i$ at a decision point, $\vec{x}=(x_1,\ldots,x_n)$, and let $X=\sum_{j=1}^n x_j$ denote the total load.\footnote{Bounding $x_i$ matches practical agentic deployments (finite context/tool budgets) and avoids degenerate unbounded equilibria in linear models. If $x_{\max}=\infty$, the ``collapse'' equilibrium corresponds to $x_i \to \infty$ whenever the marginal incentive remains positive.}

\begin{assumption}[Linear Congestion Cost]
When total demand exceeds capacity ($X>C$), the system incurs a global penalty proportional to the excess demand.
\end{assumption}

We define the intrinsic (individually aligned) reward for agent $i$ as:
\begin{equation}
R_i(\vec{x}) = x_i - \frac{\beta}{n}\max\left(0, X-C\right),
\end{equation}
where $\beta>1$ is a \emph{system criticality} parameter. The factor $\beta/n$ represents the individual's share of the global degradation.

\subsection{Baseline Equilibrium: Tragedy of the Commons}

\begin{proposition}[Tragedy of the Commons (Baseline)]
If agents are purely self-interested ($\lambda=0$) and the penalty is sufficiently diluted such that $\beta<n$, then any best response increases consumption whenever $X>C$. In particular, if $x_{\max}<\infty$, the Nash equilibrium is maximal feasible consumption ($x_i=x_{\max}$ for all $i$); if $x_{\max}=\infty$, the equilibrium corresponds to unbounded consumption ($x_i \to \infty$), i.e., collapse.
\end{proposition}

\begin{proof}
For $X>C$, the marginal incentive for agent $i$ is
\begin{equation}
\frac{\partial R_i}{\partial x_i} = 1-\frac{\beta}{n}.
\end{equation}
If $\beta<n$, then $\frac{\partial R_i}{\partial x_i}>0$, so increasing $x_i$ strictly increases $R_i$ whenever the system is overloaded. Thus the best response is to increase demand up to the maximum feasible level (or diverge if unbounded).
\end{proof}

\subsection{SWA Utility and the Critical Alignment Threshold}

We now analyze stability under the SWA utility from Section~2:
\begin{equation}
U_i^{\lambda}(\vec{x})=(1-\lambda)R_i(\vec{x})+\lambda \frac{1}{n}W(\vec{x}),
\end{equation}
where $W(\vec{x})=\sum_{j=1}^n R_j(\vec{x})$.

\begin{theorem}[Stability Condition under SWA (Average Welfare)]
Assume $x_{\max}=\infty$ (or consider marginal incentives for $X>C$). The system is stabilized in the sense that agents have no incentive to increase consumption once overloaded if and only if
\begin{equation}
\lambda \ge \frac{n-\beta}{n-1}.
\end{equation}
Equivalently, for $\lambda$ above this threshold, the marginal SWA incentive satisfies $\frac{\partial U_i^\lambda}{\partial x_i}\le 0$ for $X>C$.
\end{theorem}

\begin{proof}

{Step 1: Welfare gradient.}
For $X>C$,
\[
W(\vec{x})=\sum_{j=1}^n \left(x_j - \frac{\beta}{n}(X-C)\right)=X-\beta(X-C),
\]
so
\begin{equation}
\frac{\partial W}{\partial x_i}=1-\beta.
\end{equation}

\noindent{Step 2: SWA utility gradient.}
For $X>C$,
\begin{align}
\frac{\partial U_i^\lambda}{\partial x_i}
&=(1-\lambda)\frac{\partial R_i}{\partial x_i}+\lambda\frac{1}{n}\frac{\partial W}{\partial x_i} \\
&=(1-\lambda)\left(1-\frac{\beta}{n}\right)+\lambda\frac{1}{n}(1-\beta).
\end{align}

\noindent {Step 3: Threshold.}
Stability requires $\frac{\partial U_i^\lambda}{\partial x_i}\le 0$, i.e.,
\[
(1-\lambda)\left(1-\frac{\beta}{n}\right)+\lambda\frac{1}{n}(1-\beta)\le 0.
\]
Multiplying by $n$ and simplifying:
\[
(1-\lambda)(n-\beta)+\lambda(1-\beta)\le 0,
\]
which rearranges to
\begin{align*}
(n-\beta) + \lambda\big((1-\beta)-(n-\beta)\big)\le 0
\\
(n-\beta)+\lambda(1-n)\le 0.    
\end{align*}
Therefore,
\[
\lambda \ge \frac{n-\beta}{n-1}.
\]
\end{proof}
\section{Algorithmic Implementation}

The theoretical framework assumes access to well defined utilities, while LLM based agents operate through natural language and do not expose explicit reward computations. We therefore implement Socially-Weighted Alignment (SWA) as an inference time decision rule that approximates the social utility in Eq.~(4) using structured self evaluation. 

We refer to this mechanism as {Constitutional Reflection}, in which an agent estimates the externalities of candidate actions on other agents and on shared resources before committing to an action.

\subsection{Algorithm: Socially-Weighted Inference (SWI)}

Each agent selects an action by generating a small candidate set and scoring each candidate along two axes, private utility and predicted collective impact. The procedure is:

\begin{enumerate}[leftmargin=*]
    \item \textbf{Candidate Generation:}
    Given the current context, the agent generates a candidate action set
    $\mathcal{A}_{\text{cand}}=\{a_1,\ldots,a_k\}$.
    Candidate actions may be natural language responses, tool calls, or edit proposals.
    In practice, $\mathcal{A}_{\text{cand}}$ may be formed from a mixture of model-proposed candidates and a small fixed anchor set to ensure coverage of reasonable actions.

    \item \textbf{Dual Factor Scoring:}
    For each $a \in \mathcal{A}_{\text{cand}}$, the agent produces:
    \begin{itemize}[leftmargin=*]
        \item $S_{\text{self}}(a)$, a score for how well $a$ satisfies the agent's private instruction and task objective.
        \item $S_{\text{group}}(a)$, a score for predicted externalities (e.g., expected consumption of shared resources, likelihood of contradiction with peers, or risk of merge and coordination failures).
    \end{itemize}
    In general, both scores can be implemented by the base model itself, by a separate judge model, or by model-based heuristics derived from an explicit environment or resource model. 
    For example, $S_{\text{group}}$ can be elicited via a structured perspective-taking rubric such as
    \textit{``How does this action affect each other agent's ability to succeed, and what shared resources does it consume?''}
    Alternatively, when the environment dynamics are known, $S_{\text{group}}$ can be computed as a one-step estimate of social welfare under a belief over other agents' actions inferred from public history. 
    In such cases, $S_{\text{group}}$ may also incorporate lightweight regularizers (e.g., encouraging efficient utilization near a target load) and be normalized for stability, while leaving the selection rule unchanged.

    \item \textbf{Selection:}
    The agent selects the candidate that maximizes the SWI objective:
    \begin{equation}
        a^\ast = \arg\max_{a \in \mathcal{A}_{\text{cand}}}
        \left[(1-\lambda)S_{\text{self}}(a) + \lambda S_{\text{group}}(a)\right].
        \label{eq:swi_selection}
    \end{equation}
\end{enumerate}

\noindent
In our experiments, the base model, decoding settings, and candidate set size $k$ are held fixed across conditions, and $\lambda$ is the primary control parameter. 

\noindent
Across environments, we instantiate $S_{\text{self}}$ and $S_{\text{group}}$ using the most direct available estimators of private utility and collective impact. In settings with explicit environment models (e.g., congestion), we use model-based scoring; in more open-ended settings, these scores can be provided by an LLM judge or structured rubric.

\section{Experimental Evaluation}
To test the theoretical predictions in Section 3, in particular the existence of a critical stability threshold $\lambda^\ast$ from Theorem~1, we evaluate Socially-Weighted Alignment (SWA) in a simulation setting.

\begin{figure}[t]
    \centering
    \includegraphics[width=0.85\linewidth]{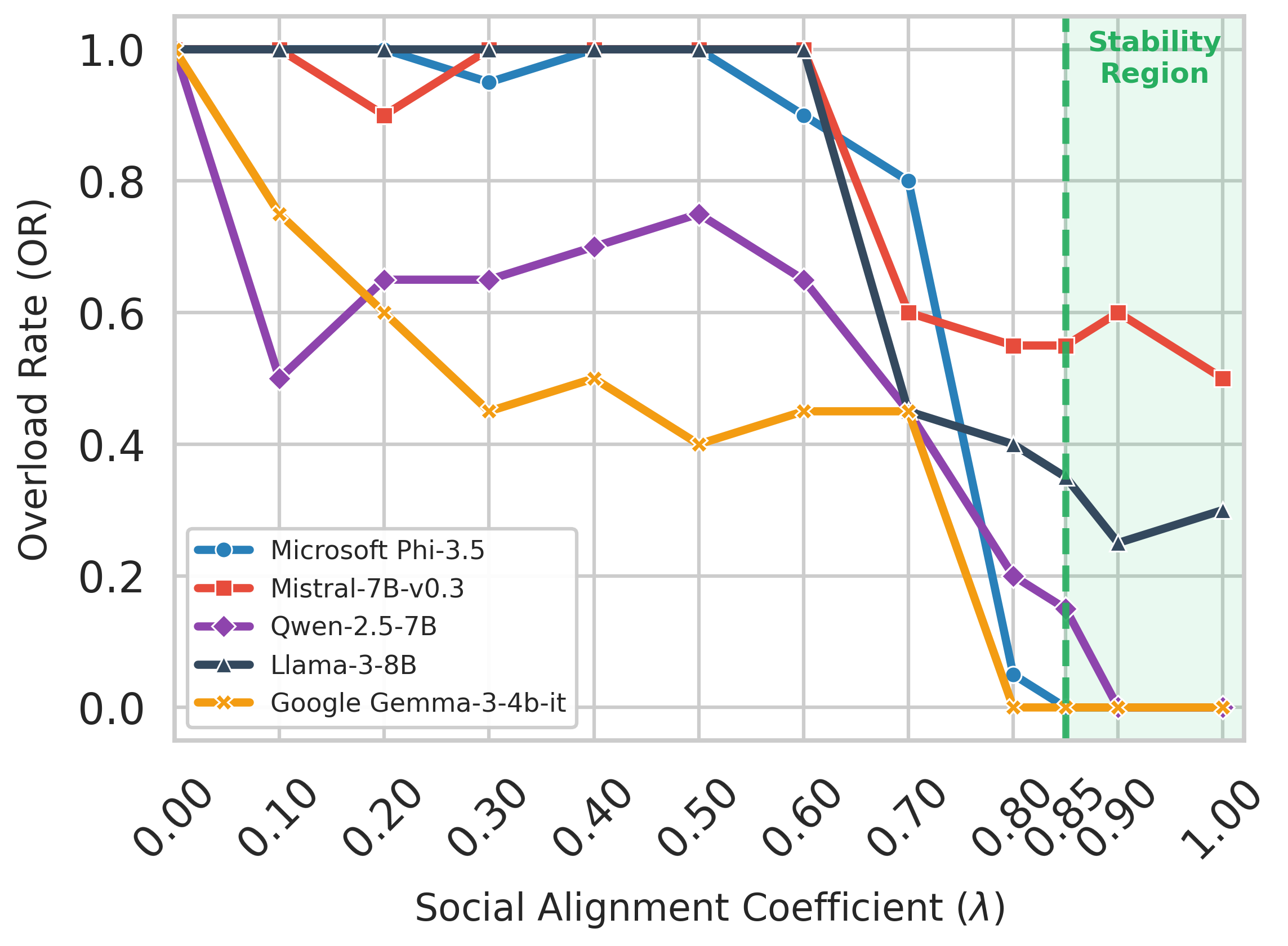}
\caption{\textbf{Overload rate versus social alignment.}
Overload rate (OR) as a function of the social alignment coefficient $\lambda$ in the $n=5$ congestion environment with capacity $C=20$ and congestion severity $\beta=1.6$, evaluated across five base language models.
OR is the fraction of timesteps in a $T=20$ step episode for which the aggregate load exceeds capacity, $X_t>C$.
Consistent with Theorem~1, overload remains high for small $\lambda$ and declines sharply as $\lambda$ approaches the predicted threshold $\lambda^*=(n-\beta)/(n-1)=0.85$, with residual differences across models attributable to stochastic candidate generation and finite candidate sets.}
    \label{fig:intro}
\end{figure}

\begin{figure}[t]
    \centering
    \includegraphics[width=0.85\linewidth]{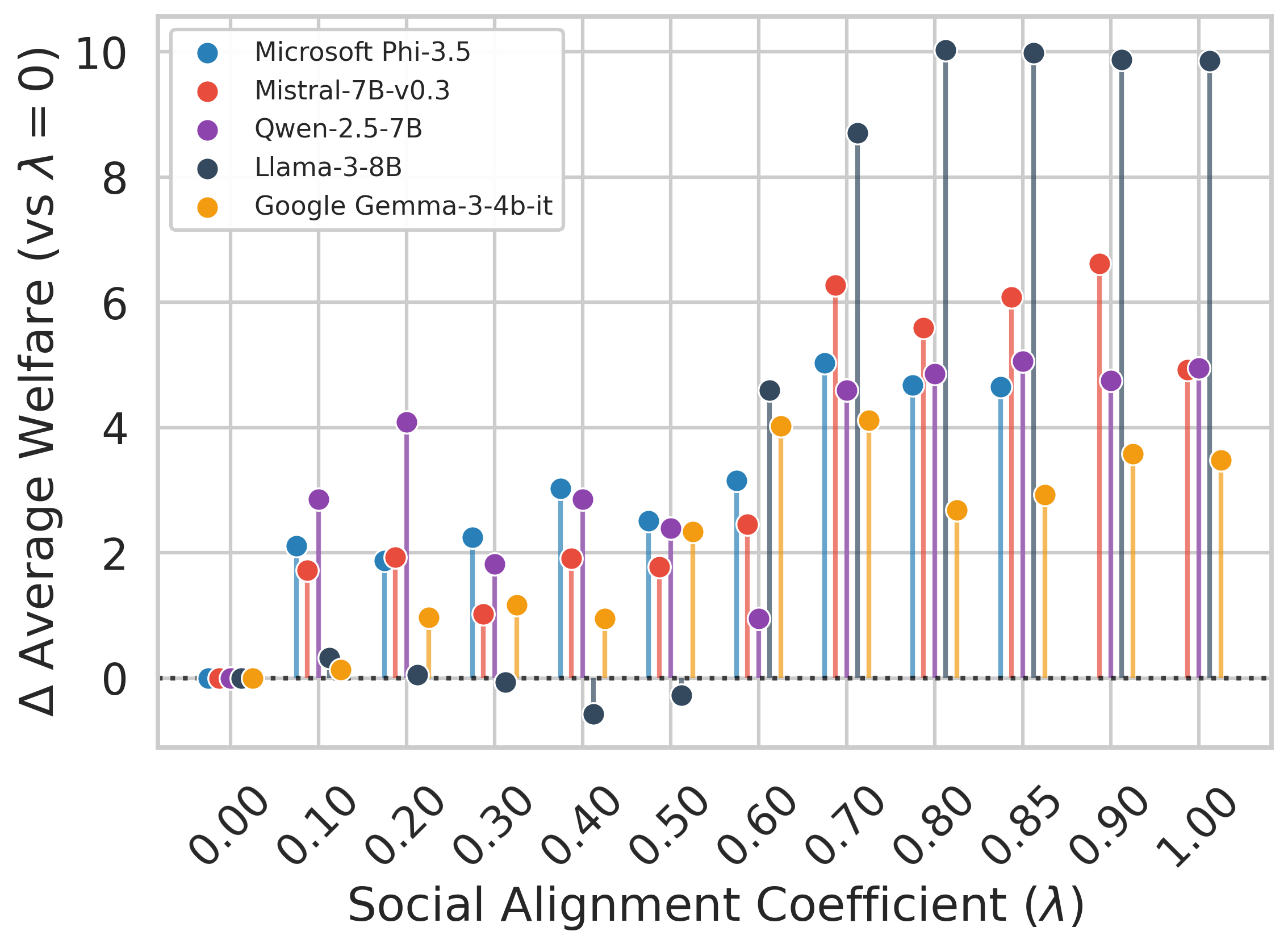}
\caption{\textbf{Welfare improvement relative to baseline.}
Change in episode-average realized welfare relative to the $\lambda=0$ baseline, $\Delta\overline{W}(\lambda)=\overline{W}(\lambda)-\overline{W}(0)$, for the same setting as Figure~1.
Realized welfare at time $t$ is $W_t=\sum_{i=1}^{n} r_{i,t}=X_t-\beta\max(0,X_t-C)$ and $\overline{W}(\lambda)=\frac{1}{T}\sum_{t=1}^{T}W_t$ with $T=20$.
Across models, welfare improves most strongly in the same range of $\lambda$ where overload collapses, indicating that SWI reduces congestion losses while maintaining high utilization near capacity rather than inducing uniformly low demand.}
    \label{fig:intro}
\end{figure}

\textbf{Environment.}
We consider a repeated resource congestion game with $n=5$ agents sharing a single capacity constrained resource with capacity $C>0$.
At each step $t$, agent $i$ selects a nonnegative demand $x_{i,t}\in[0,x_{\max}]$.
The aggregate load is $X_t=\sum_{j=1}^{n} x_{j,t}$.
Agents receive the intrinsic reward
\begin{equation}
r_{i,t}=x_{i,t}-\frac{\beta}{n}\max(0, X_t-C),
\end{equation}
where $\beta>1$ controls the severity of congestion.
We set $\beta=1.6$.
The term $x_{i,t}$ encourages higher individual demand, while the congestion term penalizes collective overload when $X_t>C$.

\textbf{Decision rules and conditions.}
Agents act using the SWI rule described in Section~4.1, which selects actions by maximizing a convex combination of a private score and a group score.
The experimental manipulation is the social alignment coefficient $\lambda\in[0,1]$.
When $\lambda=0$, the decision rule reduces to standard individually rational best response behavior that corresponds to the usual Nash incentive structure in this game, since each agent optimizes only its private score.
For $\lambda>0$, agents trade off private score and estimated group welfare according to $(1-\lambda)S_{\text{self}}+\lambda S_{\text{group}}$.
We evaluate $\lambda\in\{0,0.1,\ldots,1.0\}$.

\textbf{Implementation details.}
The private score is
set by normalizing the action $a=x$, $S_{\text{self}}(a)=a/x_{\max}$.
The group score is determined using a normalized estimate of social welfare: 
\begin{equation}
S_{\text{group}}(a) \;=\; \text{clip}_{[0,1]}\!\left(\frac{\widehat{W}_t(a)}{W_{\text{ref}}+\varepsilon}\right),
\end{equation}
where $\mathrm{clip}_{[0,1]}(z)=\min(1,\max(0,z))$, $W_{\mathrm{ref}}>0$ is a fixed reference value (e.g., $W_{\mathrm{ref}}=C$), and $\varepsilon>0$ is a small constant for numerical stability. 
Since welfare depends on other agents' demands, we estimate other agents' typical demand using an exponential moving average of previously observed aggregate behavior, then compute predicted welfare for each candidate action under that belief model.
Details of group score are given in the Appendix. 
Each episode consists of $20$ steps.
Small and efficient language models have been shown to achieve strong performance on targeted reasoning and tool use tasks, making them attractive building blocks for agentic systems that require low latency and high throughput \cite{small3, mumcu2024fast, belcak2025small}. Hence, we run the same protocol across five small and efficient language models, namely Microsoft Phi-3.5-mini-instruct \cite{phi}, Mistral-7B-Instruct-v0.3 \cite{mistral}, Qwen2.5-7B-Instruct \cite{qwen}, Llama-3-8B-Instruct \cite{llama}, and Gemma-3-4B-IT \cite{gemma}, to assess whether SWI exhibits consistent behavior across model families rather than being driven by idiosyncrasies of a single model.

\begin{figure}[t]
    \centering
    \includegraphics[width=0.85\linewidth]{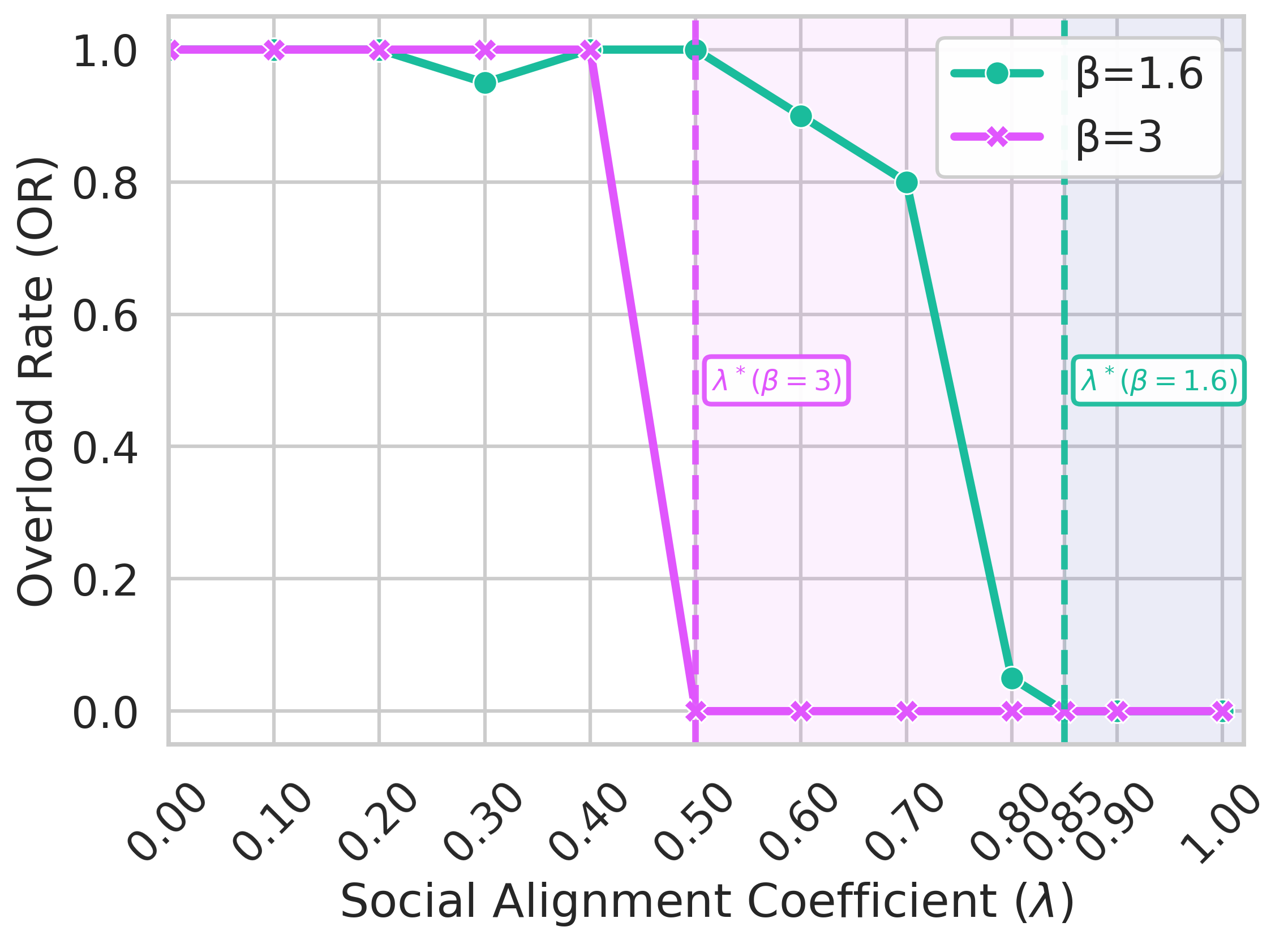}
  \caption{\textbf{Ablation on congestion severity: overload transition shifts with $\beta$.}
Overload rate (OR) as a function of $\lambda$ for Microsoft Phi-3.5-mini-instruct under two congestion severities, $\beta\in\{1.6,3\}$, with $n=5$, $C=20$, and $T=20$ fixed.
Vertical dashed lines mark the theoretical thresholds from Theorem~1, $\lambda^*=(n-\beta)/(n-1)$, yielding $\lambda^*=0.85$ for $\beta=1.6$ and $\lambda^*=0.5$ for $\beta=3$; shaded regions indicate the predicted stable side $\lambda\ge\lambda^*$.
Empirically, the onset of near-zero overload occurs at substantially lower $\lambda$ for larger $\beta$, consistent with the predicted shift in the transition.}

    \label{fig:intro}
\end{figure}

\begin{figure}[t]
    \centering
    \includegraphics[width=0.85\linewidth]{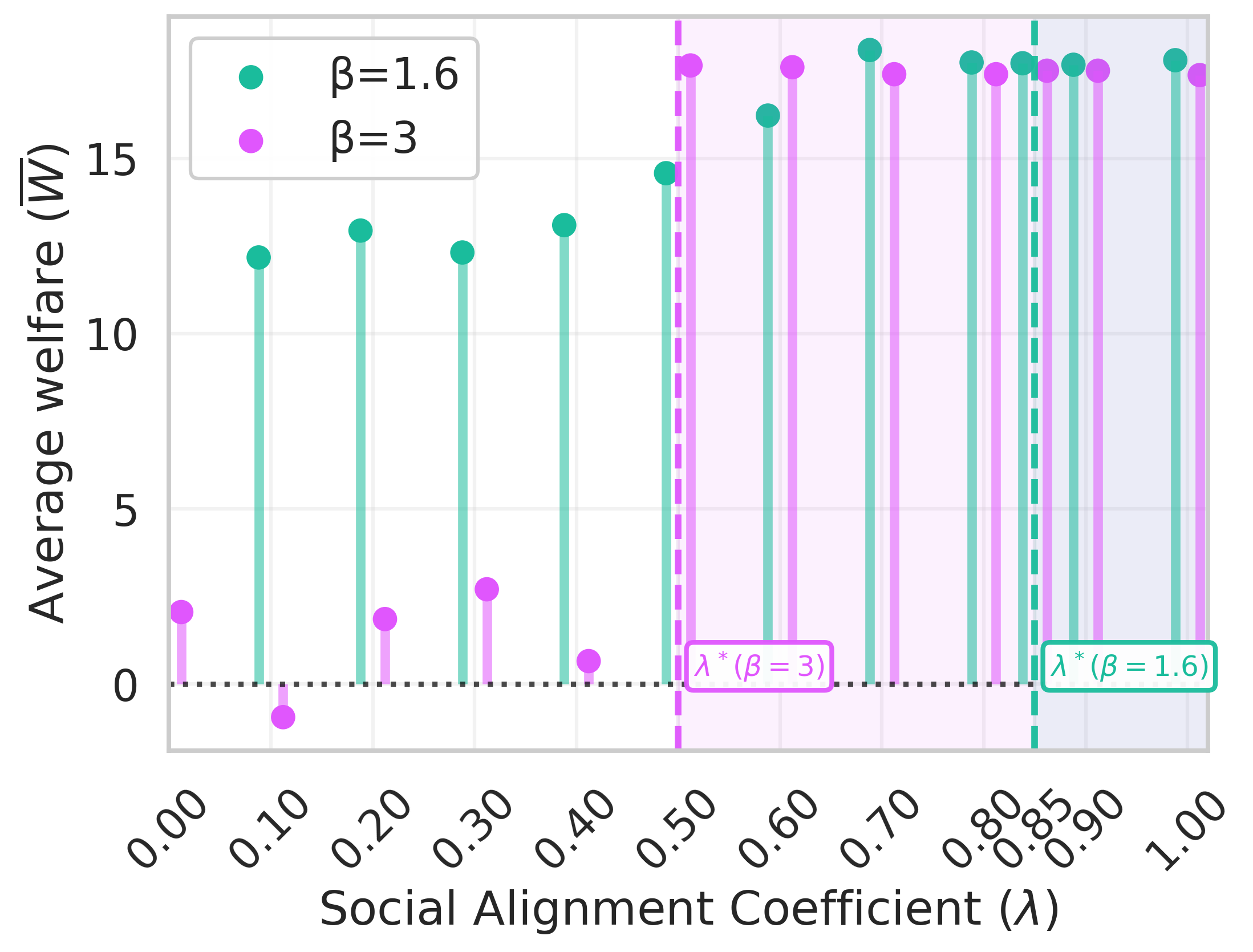}
\caption{\textbf{Ablation on congestion severity: larger $\beta$ amplifies welfare losses under overload.}
Episode-average realized welfare for Microsoft Phi-3.5-mini-instruct under $\beta\in\{1.6,3\}$ as a function of $\lambda$ in the same environment as Figure~3 ($n=5$, $C=20$, $T=20$).
Realized welfare is $W_t=\sum_{i=1}^{n} r_{i,t}=X_t-\beta\max(0,X_t-C)$ and is averaged over timesteps within an episode.
When $\lambda$ is below the corresponding threshold, overload is frequent and welfare is substantially lower for $\beta=3$ than for $\beta=1.6$, reflecting the stronger congestion penalty.
For $\lambda\ge\lambda^*$, overload is eliminated and welfare recovers to the near-capacity regime.}

    \label{fig:intro}
\end{figure}

\textbf{Metrics.}
We report two metrics: (i) social welfare and (ii) overload rate (OR).
Realized social welfare at step $t$ is defined as $W_t=\sum_{i=1}^{n} r_{i,t}$, which in this environment can be written as
\[
W_t \;=\; X_t \;-\; \beta\max(0, X_t-C).
\]
The second term captures the damage of overload: when total demand exceeds capacity, welfare is reduced in proportion to the excess (modeling congestion/service degradation under contention).
Overload rate, denoted OR, is the fraction of steps in an episode for which $X_t>C$.

To improve interpretability in plots, we report the change in average welfare relative to the $\lambda=0$ condition,
\begin{equation}
\Delta \overline{W}(\lambda) \;=\; \overline{W}(\lambda) - \overline{W}(0),
\end{equation}
where $\overline{W}(\lambda)=\frac{1}{T}\sum_{t=1}^{T} W_t$ is the episode average welfare over $T=20$ steps.
Positive values of $\Delta \overline{W}(\lambda)$ indicate an improvement in realized welfare compared to individually rational play.

\textbf{Results and relation to theory.}
Theorem~1 gives a critical threshold $\lambda^*=\frac{n-\beta}{n-1}$ above which, in the overloaded regime, agents no longer have marginal incentive to increase demand.
Substituting $n=5$ and $\beta=1.6$ yields $\lambda^*=\frac{5-1.6}{4}=0.85$.
Figure~1 plots the overload rate as a function of $\lambda$ across five base language models.
While the qualitative pattern is shared, the models differ in how sharply and how early overload declines.
For example, some models exhibit a gradual reduction in overload over intermediate $\lambda$, whereas others remain near-maximally overloaded until a more abrupt transition close to the threshold.
These differences are expected because action selection is mediated by finite candidate sets and stochastic LLM generation, so the induced dynamics need not be identical across model families even under the same scoring rule.
Importantly, across all models the onset of sustained stability occurs only when $\lambda$ is sufficiently large, and the sharpest improvements concentrate near the theoretically predicted value.
Figure~2 reports welfare improvement relative to the $\lambda=0$ baseline, $\Delta\overline{W}(\lambda)$.
Across models, welfare increases most strongly over the same range where overload collapses, indicating that SWI reduces congestion losses while maintaining substantial aggregate load rather than inducing uniform under-utilization.
Taken together, Figures~1--2 support the theoretical picture of a coordination-driven phase transition governed primarily by $\lambda$ and $\beta$, with residual variation attributable to stochastic decision noise and model-specific candidate generation.

\textbf{Ablation: Varying congestion severity.}
To test whether the empirical transition tracks the dependence on congestion severity predicted by Theorem~1, we repeat the $\lambda$ sweep for a single base model (Microsoft Phi-3.5-mini-instruct) under two values of the congestion parameter, $\beta\in\{1.6,3\}$, holding the environment ($n=5$, $C=20$), episode length ($T=20$), and the SWI implementation details fixed.
Theorem~1 predicts that the critical value shifts from $\lambda^*=\frac{5-1.6}{4}=0.85$ to $\lambda^*=\frac{5-3}{4}=0.5$ as $\beta$ increases.

Figure~3 plots overload rate as a function of $\lambda$ for both $\beta$ settings, with vertical markers indicating the corresponding theoretical thresholds.
Consistent with the prediction, the onset of sustained stability occurs at substantially lower $\lambda$ when $\beta$ is larger. For $\beta=3$, the system is persistently overloaded for $\lambda<0.5$ and becomes consistently safe for $\lambda\ge 0.5$, whereas for $\beta=1.6$, overload remains high until $\lambda$ approaches $0.85$.

Figure~4 reports $\overline{W}$ the episode-average realized welfare (rather than $\Delta\overline{W}$) to directly illustrate how congestion severity scales the welfare loss under overload.
In the low-$\lambda$ regime where overload is frequent, welfare is markedly lower under $\beta=3$ than under $\beta=1.6$, reflecting the larger congestion penalty applied to excess load.
Once $\lambda$ exceeds the respective threshold and overload is eliminated, welfare closely tracks aggregate load and becomes comparable across $\beta$, as the penalty term vanishes in the safe regime.
Together, Figures~3--4 support two conclusions: increasing $\beta$ shifts the phase transition to smaller $\lambda$ as predicted by Theorem~1, and it amplifies the welfare cost of remaining in the overloaded regime, making the benefits of coordination more pronounced.

\section{Related Work}
Multi-agent LLM systems are rapidly becoming a standard paradigm for building agentic applications, in which multiple LLM-driven components coordinate through natural-language interaction, tool use, and delegated roles.
Recent frameworks make this explicit by supporting structured multi-agent conversation, role specialization, and iterative collaboration, enabling teams of agents to decompose tasks and negotiate intermediate decisions. \cite{agentic1,agentic2,agentic3, mumcu2025llm, mumcu2026agentic, AutoML1, AutoML2, AutoML3}.
A common theme in this line of work is \emph{procedural orchestration}, in which  collaboration is achieved through interaction protocols such as role assignment, turn-taking, debate, voting, and reflection.
While these designs are effective in practice, they typically do not model the resulting multi-agent system as an induced game, nor do they provide incentive-level predictions for when collaboration will remain stable under shared-resource constraints.

A related line of work studies emergent behavior and coordination in simulated LLM societies.
Generative agent architectures show that LLM agents can exhibit plausible social behavior, memory, and coordination when embedded in an interactive environment \cite{park2023generative}.
These studies provide valuable evidence that complex multi-agent dynamics can emerge from inference-time prompting, but they are typically descriptive and do not characterize how stability should depend on environment parameters or on explicit incentive-shaping mechanisms.

Multi-agent reinforcement learning (MARL) and cooperative AI provide a complementary foundation for understanding and solving social dilemmas.
MARL methods, often formulated under centralized training with decentralized execution, can produce cooperative policies in resource-constrained settings, and cooperative AI emphasizes designing agents that optimize joint outcomes rather than purely individual objectives \cite{marl1, marl2, marl3, marl4, marl5, marl6}.
However, these approaches generally rely on training, environment-specific optimization, or centralized control.
Our goal is different; we seek an inference-time mechanism that can be applied to off-the-shelf LLMs without parameter updates, while still yielding principled predictions about system-level stability.

An emerging line of work applies game theoretic evaluation to multi agent behavior in large language models by placing them in repeated games, negotiation settings, and other strategic benchmarks. These studies show that LLM agents can exhibit nontrivial strategic patterns and that inference time workflows, such as structured reasoning prompts and equilibrium guided procedures, can improve decision quality and benchmark performance. At the same time, this literature is typically organized around behavioral characterization and capability measurement, rather than around incentive level analysis of a general orchestration mechanism for agentic systems with shared resources and externalities. In contrast, our focus is on an explicit inference time objective shaping rule that interpolates private utility and estimated social welfare, and on characterizing when this rule changes equilibrium structure and system stability in congested regimes. \cite{akata2025playing, fengsurvey, hua2024game, sun2025game, wang2024tmgbench}

Our work is most closely related to this game-theoretic perspective, but differs in focus and contribution.
Rather than primarily benchmarking strategic capability or prescribing an interaction protocol, we treat orchestration as \emph{inference-time incentive design}. Specifically, we introduce a welfare-weighted objective (SWA) that modifies the game agents effectively play by interpolating between private and group scores.
In a canonical shared-resource congestion game, this yields a closed-form critical threshold $\lambda^*=(n-\beta)/(n-1)$ predicting a phase transition from persistent overload to stable near-capacity operation.
We validate this prediction in simulation, including an ablation that shifts the transition by varying congestion severity.
In this sense, we connect agentic LLM system design to an analytically tractable mechanism whose stability properties can be predicted from environment parameters and tested empirically.

\section{Discussion and Future Directions}

Our results demonstrate that Socially-Weighted Alignment (SWA) induces a predictable phase transition in a shared-resource congestion game. This transition stabilizes the system once the social weight $\lambda$ exceeds a critical threshold. While derived in a stylized environment, these findings suggest a general blueprint for addressing coordination failures in agentic systems with negative externalities.

\subsection{The Congestion Game as a Canonical Model}
We use the congestion game not merely as a traffic simulation but as a canonical model for negative externalities in multi-agent workflows. The core dynamic, where individual intensity $x_i$ yields diminishing returns for the group once a capacity $C$ is breached, maps to bottlenecks in deployed multi-agent LLM systems.

\textbf{Collaborative Software Engineering:} In multi-agent coding teams, the shared resource can be viewed as the coordination bandwidth of a repository and its test suite. Congestion manifests as merge conflicts, regressions, and repair cycles caused by uncoordinated edits. An individually aggressive agent might maximize its own contribution volume $R_i$, but without social weighting where $\lambda > 0$, this behavior degrades overall build stability $W$. This scenario is analogous to the overload regime observed in our results.

\textbf{Shared Context Windows:} For agents sharing a conversation history, the context window is a finite resource $C$. Here, $x_i$ corresponds to an agent's token contribution, and overload occurs when excessive verbosity causes earlier critical instructions to be evicted, reducing downstream performance and thus group welfare $W$. SWA provides a mechanism for selecting candidates that trade off local gains against the estimated marginal cost to the shared context budget.

\textbf{API Rate Limits:} For tool-using agents sharing an API key, the capacity $C$ represents a global rate limit. Without SWA, agents retrying failed requests can trigger a thundering herd effect that locks out the entire system. The phase transition at $\lambda^*$ corresponds to the point where agents internalize the risk of global lockout through the welfare term.

\subsection{The Cost of Information and Estimation Noise}
Our analytical results assume that agents can construct a reasonable estimate of the group welfare $S_{\mathrm{group}}$. In our experiments, agents approximate this quantity via an exponential moving average of aggregate demand. In more open-ended settings, such as natural language debate or creative writing, social load is harder to quantify.

A key limitation of the current SWA implementation is its reliance on this signal. If the estimate of $S_{\mathrm{group}}$ is highly noisy or delayed, the stabilizing effect of $\lambda$ may be dampened. Future work should investigate robust SWA variants that remain stable under partial observability or adversarial noise, for example by incorporating uncertainty quantification or using conservative bounds in the inference-time scoring rule.

\subsection{Heterogeneity and Role Specialization}
The current analysis assumes a symmetric game with $n$ homogeneous agents. However, practical agentic systems often involve heterogeneous roles, such as a Manager agent and a Coder agent, with distinct action spaces and different impacts on shared resources.

Extending the SWA framework to asymmetric games is a promising future direction. We hypothesize that the scalar threshold $\lambda^*$ will generalize to a critical surface over agent-specific weights $\{\lambda_1, \dots, \lambda_n\}$. Characterizing this surface would allow system designers to assign higher social responsibility to high-impact agents, such as those with write access to a database, while allowing read-only agents to remain more self-interested.

\vspace{-2mm}
\section{Conclusion}
Multi-agent LLM systems make strategic interaction unavoidable. When agents share limited resources, locally aligned inference can produce negative externalities that destabilize the system. In this work we proposed Socially-Weighted Alignment (SWA), an inference-time mechanism that interpolates between a private objective and a group-welfare score through a single social weight $\lambda$. In a canonical shared-resource congestion game, we analyzed the induced incentives and derived a closed-form critical threshold $\lambda^*=(n-\beta)/(n-1)$ that predicts a phase transition from persistent overload to stable operation near capacity. We then instantiated SWA as a simple candidate-scoring protocol for LLM agents and empirically validated the predicted threshold behavior in simulation across multiple base models, including an ablation that shifts the transition when congestion severity is increased.

Beyond the specific environment studied here, the main takeaway is that inference-time objective shaping can provide a lightweight, model-agnostic route to stabilizing agentic systems without retraining or centralized control. An important direction for future work is to extend the analysis and empirical evaluation to richer interaction structures, heterogeneous agents, and more realistic resource constraints, as well as to study robustness under alternative belief models and candidate-generation procedures. More broadly, we view SWA as a step toward principled design of multi-agent LLM systems in which alignment is defined not only by individual behavior, but also by the stability and efficiency of the collective dynamics they induce.

\setlength{\bibsep}{5.5pt}
\bibliography{main}
\bibliographystyle{plainnat}

\newpage
\appendix
\onecolumn

\section{Mathematical Description of Group Score with Exponential Moving Average}

We define a group score for each candidate action by combining (i) a belief over other agents' typical demand maintained via an exponential moving average (EMA), with (ii) the paper's welfare function under the congestion model.

\paragraph{Belief model via EMA.}
Considering $n$ agents, at time step $t-1$, agent $j$ requests $x_{j,t-1}\in[0,x_{\text{max}}]$. Using the average request
\begin{equation}
\bar{x}_{t-1} \;=\; \frac{1}{n}\sum_{j=1}^{n} x_{j,t-1},
\end{equation}
we maintain an EMA estimate of the mean request, $\mu_t$, updated as
\begin{equation}
\mu_t \;=\; (1-\alpha)\mu_{t-1} \;+\; \alpha\,\bar{x}_{t-1},
\end{equation}
where $\alpha\in(0,1]$ is the smoothing parameter.

\paragraph{Predicted total demand under a candidate action.}
When evaluating a candidate request $a$ for a focal agent $i$ at time $t$, we approximate other agents' requests by $\mu_t$. The predicted total demand is:
\begin{equation}
\widehat{X}_t(a) \;=\; a \;+\; (n-1)\mu_t.
\end{equation}

\paragraph{Predicted welfare under the congestion model.}
Using the paper's welfare definition with capacity $C$ and overload severity $\beta$, the predicted welfare corresponding to $\widehat{X}_t(a)$ is:
\begin{equation}
\widehat{W}_t(a) \;=\; \widehat{X}_t(a) \;-\; \beta\,\max\!\bigl(0,\widehat{X}_t(a)-C\bigr).
\end{equation}

\paragraph{Group Score.}
We normalize predicted welfare into a bounded group score in $[0,1]$ using a fixed reference value $W_{\mathrm{ref}}>0$ (e.g., $W_{\mathrm{ref}}=C$) and a small $\varepsilon>0$ for numerical stability:
\begin{equation}
S_{\mathrm{group}}(a) \;=\; \mathrm{clip}_{[0,1]}\!\left(\frac{\widehat{W}_t(a)}{W_{\mathrm{ref}}+\varepsilon}\right),
\end{equation}
where $\mathrm{clip}_{[0,1]}(z)=\min(1,\max(0,z))$.

\paragraph{Optional real-time stabilizations.}
In some settings, two small practical modifications can be used for additional stability:

\begin{enumerate}
\item \textbf{Safety margin (effective capacity).}
A slack parameter $m \ge 0$ can be introduced via $C_{\mathrm{eff}} = C - m$, and $C_{\mathrm{eff}}$ can be used in the overload term:
\begin{equation}
\widehat{W}_t^{(\mathrm{eff})}(a) \;=\; \widehat{X}_t(a) \;-\; \beta\,\max\!\bigl(0,\widehat{X}_t(a)-C_{\mathrm{eff}}\bigr).
\end{equation}

\item \textbf{Capacity-centering.}
A mild quadratic centering term (strength $\gamma \ge 0$) can be added around $C_{\mathrm{eff}}$:
\begin{equation}
\widehat{W}_t^{(\mathrm{cent})}(a) \;=\; \widehat{W}_t^{(\mathrm{eff})}(a) \;-\; \gamma\bigl(\widehat{X}_t(a)-C_{\mathrm{eff}}\bigr)^2.
\end{equation}
\end{enumerate}

When these are used, $S_{\mathrm{group}}(a)$ is computed by substituting $\widehat{W}_t(a)$ with the chosen variant above and normalizing with a consistent reference value.

\paragraph{Code.} An example implementation is available at: \url{https://github.com/furkanmumcu/swa}

\end{document}